\documentclass[10pt]{amsart}

\usepackage{graphicx}  
\usepackage{dcolumn}   
\usepackage{bm}        
\usepackage{amssymb}   
\usepackage{amsmath}
\usepackage{amsfonts}
\usepackage{amsthm}
\usepackage{amssymb}
\usepackage{booktabs}
\newcommand{\R}{\mathbb{R}}
\newcommand{\E}{\mathbf{E}}

\newcommand{\T}{\mathcal{T}}

\newcommand{\p}{\partial}
\newcommand{\pt}{\tilde{\p}}

\newcommand{\nn}{\nonumber}
\newcommand{\existss}{\exists\hspace{1mm}}
\newcommand{\foralll}{\forall\hspace{1mm}}
\usepackage{epigraph}
\usepackage{braket} 
\usepackage{etoolbox}
\usepackage{amscd} 

\usepackage{xcolor}
\usepackage{graphicx}

\newtheorem{teo}{Theorem}
\newtheorem{prop}[teo]{Proposition}
\newtheorem{lema}[teo]{Lemma}
\theoremstyle{definition}	 
\newtheorem{cor}[teo]{Corollary}
\theoremstyle{definition} 
\newtheorem{defin}[teo]{Definition}
\theoremstyle{definition} 
\newtheorem{nota}[teo]{Remark}
\theoremstyle{definition} 

\newtheorem{ex}[teo]{Example}

\begin{document}



\title[A new class of group entropies and information geometry]{A new class of  entropic information measures, formal group theory and \\ information geometry}


\author[M. A. Rodr\'iguez]{Miguel \'A. Rodr\'iguez}
\address{Departamento de F\'{\i}sica Te\'{o}rica, Facultad de F\'{\i}sicas, Universidad
Complutense de Madrid, 28040 -- Madrid, Spain}
\author[\'A. Romaniega Sancho]{\'Alvaro Romaniega}
\address{Instituto de Ciencias Matem\'aticas, c/ Nicol\'as Ca\-brera, No. 13--15, 28049 Madrid, Spain}
\author[P. Tempesta]{Piergiulio Tempesta}
\address{Departamento de F\'{\i}sica Te\'{o}rica, Facultad de F\'{\i}sicas, Universidad
Complutense de Madrid, 28040 -- Madrid, Spain \\ and
Instituto de Ciencias Matem\'aticas, c/ Nicol\'as Ca\-brera, No. 13--15, 28049 Madrid, Spain}
\date{June 18, 2018}

\begin{abstract}
In this work, we study generalized entropies and information geometry in a group-theoretical framework. 
We explore the conditions that ensure the existence of some natural properties and at the same time of a group-theoretical structure for a large class of entropies.
In addition, a method for defining new entropies, using previously known ones with some desired group-theoretical properties is proposed.
In the second part of this work, the information geometrical counterpart of the previous construction is examined and
a general class of divergences are proposed and studied. 
Finally, a method of constructing new divergences from  known ones is discussed; in particular, some results concerning the Riemannian structure associated with the class of divergences under
investigation are formulated. 
\end{abstract}

\maketitle

\tableofcontents

\section{\label{sec:level1}Introduction}
The aim of this paper is to introduce a new, large class of entropic information measures coming from group theory and, at the same time, to establish a novel connection between the theory of 
classical generalized entropies and information geometry.

In the last decades, a wealth of new information measures have been proposed in the literature, due to their prominent role in many different scenarios arising from physics, 
social and economical sciences, and in particular classical and quantum information theory.
The famous R\'enyi and Tsallis entropies are two possible generalizations of the well-known Shannon entropy (the equivalent in information theory of Boltzmann-Gibbs entropy in thermodynamics).
They are one-parametric functions which reproduce the Shannon's functional in a specific limit. R\'enyi's entropy was introduced in 1961 as the most general classical information measure 
satisfying \emph{additivity} with respect to the composition of independent events or systems  \cite{Ren,JizAri}. Subsequently, Tsallis entropy was introduced in 1989 as a possible generalized entropy for the 
thermodynamic treatment of long-range systems; it represents an interesting example of a \emph{non-additive} measure \cite{Tsallis1,Tbook}.

More general entropies have also been proposed. For instance, Sharma-Mittal's (SM) entropy was introduced already in 1975; it represents a bi-parametric function that generalizes both R\'enyi's and
Tsallis's ones. The SM entropy was recently studied by Nielsen and Nock, \cite{Nielsen} in the context of information theory.

Starting from 1957, Jaynes \cite{Jay1,Jay2}  established a natural and fundamental correspondence between statistical mechanics and information theory. In his vision, information theory provides an inference methodology to describe general properties of arbitrary systems on the basis of incomplete information. In particular, if the available information of a system is the set of mean values of some random variables, the lest-biased distribution compatible with those constraints is postulated to be the one which maximizes the entropy.

In the last years, a new approach to information theory, called \emph{information geometry} \cite{AmaMe,AmaIG}, has emerged. Information geometry provides a new methodology applicable to various areas including information and physical sciences. Its main objective is the investigation of the geometrical structures that can be introduced in the manifold associated with the set of probability distributions of a statistical model. In this approach, one defines a Riemannian metric in a manifold of probability distributions, together
with dually coupled affine connections. A remarkable example is the celebrated Fisher metric in statistics, which is the unique invariant metric over the manifold of probability distributions according to Chentsov's theorem, see \cite{AmaMe,AmaIG} for details. Information geometry provides new tools useful in various areas of information sciences,
such as statistical inference, quantum information theory, machine learning, convex optimization, time series analysis, etc. It is also a key tool for other areas, such as neuro-computing
(where a set of neural networks  forms  a neuro-manifold, a nonlinear system equipped with the Fisher metric, \cite{AmaMe}). Here we will focus on the concept of \emph{divergence}
(and its associated geometric structure), representing a pseudo-distance over the probability manifold.

In this work, we study generalized entropies and information geometry in a broad sense. 

In order to classify generalized entropies, we adopt here the group-theoretical approach, proposed in \cite{Temp}, based on the notion of \textit{group entropy}. 
Essentially, the approach consists in an axiomatic formulation of the concept of generalized entropy, based on early work by Shannon \cite{Shannon, Shannon2} and Khinchin \cite{Khinchin}. 
More precisely, we require that an entropic function must satisfy the first three Shannon-Khinchin (SK) axioms (continuity, maximum 
entropy principle, null-composability, see below for details). Indeed, these properties are both very natural and fundamental in information theory and statistical mechanics. However, the
additivity axiom, necessary for characterizing uniquely Boltzmann's entropy \cite{Khinchin}, is replaced by a more general axiom of composability. In other words, we impose
that, when composing two statistically independent systems, the entropy of the compound system be a function of the entropy of the component systems only: $S(A\cup B)= \Phi(S(A),S(B))$. This property
is crucially required in all the space of probability distributions representing the macro-states of the system. Mathematically, the rule governing the composition process is essentially
a formal group law, which defines the group-theoretical structure associated with the given entropy in the probability space. When a function satisfies the (SK) axioms jointly with the composability one, 
we shall talk about a group entropy.

The many properties of group entropies ensure the possibility of interpreting them as good information measures in the context of the geometric theory of information, as we shall see in detail, 
and in a broader perspective in statistical mechanics.

Our main result in the theory of generalized entropies is a theorem that allows us to construct a new, 
large class of group entropies possessing relevant mathematical properties. Indeed, we shall propose a mathematical mechanism that permits to generate a ``new'' group entropy from an ``old'' one;
both of them share the same formal group law as the composition law. More generally, this idea can also be implemented by combining several entropic functions at the same time in a way that generates
new entropies with a prescribed group-theoretical structure.

At the same time, we wish to remark that the connections between information geometry and generalized entropies are abundant and well known in the literature.
First of all, the ``associated divergence'' for the Shannon-Boltzmann-Gibbs entropy, namely the  Kullback-Leibler (KL) divergence (also called Shannon relative entropy),
gives rise to a geometric structure of our statistical manifold based on the Fisher metric. We remind that Shannon's relative entropy also comes from another important class of divergences,
the Bregman divergence (\cite{Breg}), associated with exponential families, widely studied in information geometry. Furthermore, if the exponential families are extended using \emph{generalized logarithms},
 Tsallis entropy arises, \cite{AmaQ}, as the ``potential'' of a Legendre transformation. Moreover, these potentials allow us to define a dually flat geometry.
The same construction can be repeated starting from a general deformed logarithm proposed in \cite{AmaGen}, which gives rise to the notion of $\chi$-entropy.
It is also worth-mentioning that the $\alpha$-divergences in information geometry are directly related to the $\alpha$-entropy of R\'enyi.
 The specific conditions ensuring the existence of certain crucial geometrical properties are studied. A relation among generalized entropies, information theory and statistical
physics is proposed in \cite{Naudts}. Also, very recently, in the interesting paper \cite{GBP} the Fisher metric has been related with the 
universal-group entropy, a trace-form class of entropic functions which has been proposed in \cite{TempestaAP2016}.

\vspace{2mm}

The article consists of two parts closely related, one devoted to the theory of generalized entropies and the other one focused on their information-geometrical interpretation.
Let us sum up the main results of this work. In Section 2, we briefly review the main properties of group entropies.
In Section \ref{sec:Shf}, we consider a large class of entropic functions, the $S_{h,f}$-entropies and establish their group theoretical structure.
In Section \ref{sec:EntFun}, we introduce a composition mechanism that allows to construct new entropies from known ones
and we study the corresponding group-theoretical structure, see Section 5 for examples. The main results of classical information geometry are briefly reviewed in Section 6. 
In Section \ref{sec:DivFun} a mechanism for defining new semi-divergences (see below) using known ones is introduced,  and its geometry is studied. In the final Section \ref{sec:Dhf} the divergence $D_{h,f}$ is proposed and we study the associated geometry.

This work represents a foundational step of a broader research project on the relationship between group theory, generalized entropies and information geometry. 
Among many possible open research lines, we mention  two that seem to be very promising: the application of our divergences in the study of statistical
manifolds and problems of statistical inference and, at the same time, a study of the quantum version of the theory here developed, possibly in connection with the tomographic approach to quantum mechanics.

\section{Group entropies: an introduction}

In order to make the subsequent discussion self-contained, first we shall briefly summarize the main results of the theory of group entropies. The motivation is to provide an axiomatic formulation
of the theory of generalized entropies, that would allow us to connect them to information theory and statistical mechanics.

We recall that the Shannon-Khinchin axioms were proposed independently by Shannon and Khinchin as properties characterizing uniquely the mathematical form of Boltzmann's entropy as a function 
$S(p_1,\ldots,p_W)$ in the space of probability distributions. They can be stated as follows.

(SK1) (Continuity). The function $S(p_1,\ldots,p_W)$ is continuous with respect to all its arguments

(SK2) (Maximum principle).  The function $S(p_1,\ldots,p_W)$ takes its maximum value for the uniform distribution $p_i=1/W$, $i=1,\ldots,W$.

(SK3) (Expansibility). Adding an impossible event to a probability distribution does not change its entropy: $S(p_1,\ldots,p_W,0)=S(p_1,\ldots,p_W)$.

(SK4) (Additivity). Given two subsystems $A$, $B$ of a statistical system,
\[
S(A \cup B)=S(A)+S(B\mid A).
\]
It is clear that, by relaxing one of these axioms, new possibilities arise. The first three axioms are very natural properties, which are crucial both in information theory and statistical mechanics.
 Let $\mathcal{P}_W$ denote the set of all probability distributions of the form $p=(p_1,\ldots,p_W)$, $p_i\geq 0$, $\sum_i p_1 =1$.
\begin{defin}
A function $S: \mathcal{P}_W \to \mathbb{R}_{\geq0}$ that satisfies the axioms (SK1)--(SK3) will be said to be an entropic function (or a generalized entropy).
\end{defin}

Instead, the additivity axiom (SK4) has been replaced in \cite{TempestaAP2016}, \cite{TempestaPRA2016} by a more general requirement, called the \textit{composability axiom}. 
Precisely, given two \textit{statistically independent} systems $A$ and $B$, we require that the entropy satisfies the composition law
\begin{equation}
S(A\cup B)= \Phi(S(A), S(B))
\end{equation}
where $\Phi$ should fulfill some properties in order to have a physical meaning:
\begin{enumerate}
\item \label{C} $\Phi(x,y)=\Phi(y,x)$ ({\it commutativity})
\item  $\Phi(x,0)=x$ ({\it identity)}
\item  \label{A} $\Phi(x,\Phi(y,z))=\Phi(\Phi(x,y),z)$ ({\it associativity)}
\end{enumerate}
 These properties are necessary to ensure that a given entropy may be suitable for information-theoretical and
thermodynamic purposes. Indeed, it should be symmetric with respect
to the exchange of the labels $A$ and $B$ (i.e., commutativity). Furthermore, if we compose a given
system with another system in a state with zero entropy, the total entropy should
coincide with that of the given system (i.e., identity). Finally, the composability of more than two
independent systems in an associative way is crucial to ensure path independence in the composition.

If the composability axiom is satisfied in full generality, namely for all possible distributions defined in the space $\mathcal{P}_W$, the corresponding entropy will be said to be 
\emph{strictly composable}.
If it is satisfied on the uniform distribution only, the entropy is \textit{weakly composable}.
\begin{defin}
A group entropy is an entropy that satisfies the first three Shannon-Khinchin axioms and is strictly composable.
\end{defin}

\section[\texorpdfstring{(h,f)-entropies}
                        {Math symbols sum, integral}]{${S_{h,f}}$ entropies and group entropies}\label{sec:Shf}

In this Section, we shall study a class of entropic functionals, introduced in \cite{Salic}, called the $S_{h,f}$-entropies. Their main interest is of a theoretical nature, 
since very many entropies belong
 to this family. Due to their relevance in classical and quantum information theory, we shall discuss here their group-theoretical interpretation and clarify under which conditions 
these functionals  are group entropies.

\subsection{$S_{h,f}$ entropies and properties}
 We define  the $(h,f)$-entropy class as follows.
\begin{defin} \cite{Salic}, \cite{Zoz} \label{def:Sgfg} Let $f:[0,1]\to \mathbb{R}_{\ge 0}$ be a continuous, concave (respect. convex) function and $h$ be a continuous and increasing 
(respect. decreasing) map. Then, for $p\in \mathcal{P}_W$, the function
\begin{equation} \label{Shf}
S_{h,f}(p):=h\left(\sum_{i=0}^W f(p_i)\right),
\end{equation}
where $h(f(1))=0$ and $f(0)=0$, will be called the $S_{h,f}$ entropy.
\end{defin}
\begin{nota}
A huge clase of entropies belong to the class \eqref{Shf}: among them, the well known entropies of Shannon, R\'enyi,  Tsallis, Kaniadakis, Borges-Roditi and Sharma-Mittal.
\end{nota}
\begin{teo} \label{Teo:Sgfg} Let be a $S_{h,f}$ entropy. Then

a) $S_{h,f}$ is an entropic function. 

b) If $f$ is strictly concave (respect. convex) with $h$ strictly increasing (respect. decreasing), then $S_{h,f}(p)=0$  iff $p_i=1$ 
for some $i\in \{1,\ldots,W \}$  and $S_{h,f}(p)$ reaches its maximun over the uniform distribution only.
\end{teo}
\begin{proof} a) The axioms (SK1) and (SK3) are straightforwardly satisfied, because $S_{h,f}$ is a composition of continuous maps and $f(0)=0$. For proving non-negativity of $S_{h,f}$, 
observe that being $f$ concave, if $x\in[0,1]$ we have $$f(1)x=f(1)x+f(0)(1-x)\leq f(x).$$ Also, $f(1)=f(1)\sum_i p_i \leq\sum_i f(p_i)$. Then, 

\begin{equation} \label{cero}
0=h(f(1))\leq h\left(\sum_i f(p_i)\right).
\end{equation} 

If $f$ is convex and $h$ is decreasing, the last equality is still valid.

Concerning (SK2), when $f$ is concave, by Jensen's inequality we have 
\begin{equation} \label{Jensen}
\sum_{i=1}^W f(p_i)\leq Wf\left(\sum_{i=1}^W p_i\frac{1}{W}\right)=\sum_{i=1}^{W} f\left(\frac{1}{W}\right) .  
\end{equation}
Thus, $$S_{h,f}(p)=h\left(\sum_i f(p_i)\right)\leq h\left(\sum_i f\left(\frac{1}{W}\right)\right) = S_{h,f}(\bar{p}) \ ,$$
where $\bar{p}$ is the uniform distribution. As before, if $f$ is convex and $h$ is decreasing, the last equality still holds. 

b) If $f$ is strictly concave and $h$ strictly monotone, then inequality \eqref{cero} converts into an equality only for the certainty distribution $p_i=1$ for some $i\in \{1,\ldots,W\}$. 

Also, note that the inequality \eqref{Jensen} is now strict, which implies that the uniform distribution, being a maximum one, is the unique one.
Analogous considerations hold for the case of $f$ strictly convex and $h$ strictly decreasing.

\end{proof}

\begin{teo} \label{Teo:CompShf}
Let $S_{h,f}$ be an $(h,f)$-entropy and $A$ and $B$ two statistically independent systems, characterized by two independent distributions $p\in \mathcal{P}_W$, $q\in\mathcal{P}_{W'}$.
Assume that the relation
\begin{equation}\label{form:psi}
\sum_{i,j} f(p_i q_j)=\chi\left(\sum_i f(p_i),\sum_j f(q_j)\right)
\end{equation}
is satisfied for a certain function $\chi(x,y)$. Then the function $\Phi$ given by
\[
\Phi(x,y)=h\left(\chi(h^{-1} (x),h^{-1}(y))\right)
\]
determines the composition law for the entropy $S_{h,f}$:
\begin{equation}\label{clshf}
S_{h,f}(A\cup B) = \Phi (S_{h,f}(A),S_{h,f}(B)) \ .
\end{equation}
Moreover, $S_{h,f}$ will be strictly composable if and only if $\chi$ is associative, commutative and $\chi(x, f(1))=x$.
\end{teo}
\begin{proof} Given a $S_{h,f}(p)$ entropy and two independent distributions $p$, $q$, we obtain
\begin{eqnarray}
\nn S_{h,f}(p\cdot q)&=&h\bigg(\sum_{i,j} f(p_iq_j)\bigg)=h\bigg(\chi\bigg(\sum_i f(p_i),\sum_j f(q_j)\bigg)\bigg) \\
&=& h\bigg(\chi\bigg(h^{-1}(S_{h,f}(p)), h^{-1}(S_{h,f}(q))\bigg)\bigg) \ .
\end{eqnarray}
Therefore, if we introduce the function $\Phi(x,y):=h\left(\chi(h^{-1} (x),h^{-1}(y))\right)$ we get immediately the composition law \eqref{clshf}.

Let us prove the strict composability of $S_{h,f}$. It is easy to ascertain that the commutativity property: $\Phi(x,y)=\Phi(y,x)$ is ensured by the commutativity of $\chi(x,y)$.
The associativity property is also checked
\begin{eqnarray}
\nn \Phi(x,\Phi(y,z))&=&h\big(\chi(h^{-1}(x),h^{-1}(\Phi(y,z))\big)=h\big(\chi(h^{-1}(x),\chi(h^{-1} (y), h^{-1}(z))\big) \\
\nn &=&h\big(\chi(\chi(h^{-1}(x),h^{-1} (y)),h^{-1}(z)\big) = h\big(\chi(h^{-1}(\Phi(x,y)),h^{-1}(z)\big) =\\
&=&\Phi(\Phi(x,y),z) .
\end{eqnarray}
Finally,
\[
\Phi(x,0)=h\left(\chi(h^{-1} (x),h^{-1}(0))\right),
\]
but $h^{-1}(0)=f(1)$, so $\Phi(x,0)=h\left(\chi(h^{-1} (x),f(1))\right)$. Now, assuming $\chi(y,f(1))=y$, we get immediately $\Phi(x,0)=x$.

\end{proof}
\begin{nota}
If  in the previous theorem we assume that $f(1)=0$, then we have that $\Phi$ is a group law if and only if $\chi$ is. The condition $f(1)=0$ is satisfied in many important cases, as the
entropies of Boltzmann-Gibbs, Tsallis, Kaniadakis, etc.
\end{nota}
\begin{nota}
In Theorem 1 of \cite{EncTemp} the following result has been proved: suppose to have a trace-form entropy $S=\sum_{i}f(p_i)$ where $f$ is a function  of class $C^{2}((0,1)) \cap C^{1}([0,1])$, 
with $f'$ not vanishing in any open interval, and assume that the condition \eqref{form:psi} holds with $\chi$ of class $C^{1}$, then $S$ coincides with the (trace-form version of)
Tsallis entropy  $S_{q}= \sum_i p_i \ln_{q} \frac{1}{p_i}$, where $\ln_{q}(x)= \frac{x^{1-q}-1}{1-q}$.  
Under suitable hypotheses, the problem of composability of $S_{h,f}$ was also consider in Theorem 2 of \cite{EncTemp}, obtaining $f(t)=at+bt^q$. In the previous
Theorem \ref{Teo:CompShf} these additional hypotheses are not assumed.
\end{nota}

\section{New entropy functions from formal groups} \label{sec:EntFun}
The aim of this section is to present a new class of entropic information measures by means of formal group theory. Our main result is a theorem that allows to construct iteratively new entropic functions
from old ones (for instance, $S_{h,f}$ studied above), sharing the same group theoretical structure.

Let us first define the following partial order on $\R^m$. We shall say that
\begin{equation} \label{PO}
\textbf{x}\leq \textbf{y} \quad \text{iff} \quad x_i\leq y_i \quad \foralll i\in\{1,..,m\} .
\end{equation}
It is not difficult to show the following result.
\begin{prop}\label{prop:zetaS} Let $\{S_j\}_{j=1}^m$ be a family of entropic functions  and let $\zeta: \Omega\subset \R^m\to\mathbb{R}$, with $\R_{\ge0}^n\subset \Omega$, be a
continuous function such that its restriction to $\R_{\ge0}^n$ has range $\R_{\ge0}$; also assume that $\zeta$ is increasing with respect to the partial order \eqref{PO}:
$\zeta(\boldsymbol{x})\leq\zeta(\boldsymbol{y})$ whenever $\boldsymbol{x}\leq \boldsymbol{y}$.
Then, $S(p):=\zeta\left((S_i(p))_{i=1}^m\right)$ is a new entropic function.
\end{prop}
\begin{ex}[Polynomial composition of entropies] As a special realization of the previous result, we can define the following multi-parametric polynomial function:
$$S(p)=\sum_{i=1}^n\sum_{i_1+...i_n=i}\alpha^i_{i_1...i_n}\prod_{j=1}^m(S_j(p))^{i_j} \, $$
where $\alpha^i_{i_1...i_n}\ge 0$, which is still an entropic function. Note that $S(p)$ is not a $S_{h,f}$ entropy.
\end{ex}

In order to study the composability properties of entropy functions, let us prove a technical lemma.
\begin{lema} \label{Lema:Per}Let $\Phi$ be a group law. Then,
\begin{equation}
\Phi(\Phi(x_1,x_2),\Phi(x_3,x_4))=
\Phi(\Phi(x_{\sigma(1)},x_{\sigma(2)}),\Phi(x_{\sigma(3)},x_{\sigma(4)}))
\end{equation}
for $\sigma\in S_4$ (permutation group of four elements).
\end{lema}
\begin{proof}
It is enough to prove that $x_2$ and $x_3$ can be transposed. Indeed, transpositions of (1,2) and (3,4) hold by \eqref{C} (in the $\Phi$ conditions) and transpositions (1,3), (1,4) and (2,4)
can be easily decomposed using (2,3) and (1,2) or (3,4). As any permutation can be written as the composition of transpositions, the result holds. So, let us check the transposition (2,3):
\begin{align*}
&\Phi(\Phi(x_1,x_2),\Phi(x_3,x_4))=\Phi(x_1,\Phi(x_2,\Phi(x_3,x_4)))\\
&=\Phi(x_1,\Phi(\Phi(x_2,x_3),x_4))=\Phi(x_1,\Phi(\Phi(x_3,x_2),x_4))\\
&=\Phi(x_1,\Phi(x_3,\Phi(x_2,x_4)))=\Phi(\Phi(x_1,x_3),\Phi(x_2,x_4)),
\end{align*}
where we have used (in order), $\eqref{A}$ for the external $\Phi$, \eqref{A} for the internal $\Phi$, \eqref{C} and then reverse the transpositions in the opposite way (same steps, reverse order).
\end{proof}

We are ready to generate new group entropies from old ones, by means of an iterative procedure.
\begin{defin}
Let $\Phi(x,y)$ be a group law. We introduce the function $\Phi^{2^m}:\R^{2^m}\to\R$ inductively:
\begin{eqnarray}
\Phi^{1}(x):&=&id  \qquad \text{and} \\
\Phi^{2^{m+1}}(x_1,...,x_{2^{m+1}}):&=&\Phi(\Phi^{2^m}(x_1,...,x_{2^{m}}),\Phi^{2^m}(x_{2^{m}+1},...,x_{2^{m+1}}))
\end{eqnarray}
 with $m\in\mathbb{N}=\{0,1,2,...\}$.
\end{defin}
As we have shown in Lemma \ref{Lema:Per},
\[
\Phi^4(x_1,x_2,x_3,x_4)=\Phi^4(x_{\sigma(1)},x_{\sigma(2)},x_{\sigma(3)},x_{\sigma(4)})
\]
$\foralll \sigma\in S_4$. The following result ensures that new entropy functions based on the iterated composition of group laws are indeed group entropies.

\begin{teo}\label{prop:CompZeta} Let $\xi: \mathbb{R}\to\mathbb{R}$ be a strictly increasing and continuous function. Let $\Phi(x,y)$ be a formal group law, $n=2^{m}$, $m=0,1,2,\ldots$ and
let $\zeta: \mathbb{R}^{n} \to \mathbb{R}$,
with $\zeta:=\xi\circ\Phi^{n}$, be a continuous function. Assume that $\{S_j(p)\}_{j=1}^n$ are group entropies, sharing the same composition law $\Phi(x,y)$.
Then, the new entropy $Z(p):=\zeta\left(S_1(p), \ldots,S_{n}(p)\right)$  satisfies the following composition law:
\[
Z(A\cup B)= \omega(Z(A), Z(B))
\]
with
\begin{equation} \label{omega}
\omega(x,y)=\xi(\Phi(\xi^{-1}(x),\xi^{-1}(y)))\ .
\end{equation}
Furthermore, if $\xi(0)=0$, the new entropy will satisfy the group properties.
\end{teo}
\begin{proof}Let us proceed by induction. For $m=0$, we have $\Phi^{1}=id$ and $\zeta=\xi$. If $S(p)$ is a group entropy, then $Z(p)=\zeta(S(p))$ admits the composition law
\begin{align*}
&Z(p^{AB}) = \zeta(S(p^{AB}))=\zeta(\Phi(S(p^A),S(p^B)))\\
&=\zeta(\Phi(\zeta^{-1}\circ\zeta(S(p^A)),\zeta^{-1}\circ\zeta(S(p^B))))
\end{align*}
and the result follows.

Now, let us consider the case $m+1$. Set $Z_1(p):=\Phi^{2^m}(S_1(p),...,S_{2^{m}}(p))$ and $Z_2(p):=\Phi^{2^m}(S_{2^{m}+1}(p),...,S_{2^{m+1}}(p)))$,
so that
\[
Z(p)=\xi\circ \Phi(Z_1(p), Z_2(p))\ .
\]
Then,
by the induction hypothesis (with $\xi_1=\xi_2=id$), the composition law for $Z(p)$ becomes
\begin{eqnarray*}
Z(p^{AB})&=&\xi\circ \Phi(Z_1(p^{AB}), Z_2(p^{AB})) = \xi\circ \Phi(\Phi(Z_1(p^A), Z_1(p^B)), \Phi(Z_2(p^A), Z_2(p^B))) =\\
&=&\xi\circ \Phi(\Phi(Z_1(p^A), Z_2(p^A)),\Phi(Z_1(p^B), Z_2(p^B)))\ ,
\end{eqnarray*}
where we used Lemma \ref{Lema:Per}. We finally get:
\[
Z(p^{AB})=\xi\circ \Phi(\xi^{-1}(Z(p^A)),\xi^{-1}(Z(p^B))) \ .
\]
Thus we have proved that $Z(p)$ admits the composition law given by \eqref{omega}.
\par
Let us prove now that $\omega(x,y)$ is a group law. Commutativity is evident since $\Phi$ is commutative. Concerning associativity, we observe that
\begin{eqnarray}
\nn \omega(x,\omega(y,z))&=&\xi\big(\Phi(\xi^{-1}(x), \Phi(\xi^{-1}(y),\xi^{-1}(z)))\big)= \xi\big(\Phi(\Phi(\xi^{-1}(x), \xi^{-1}(y)),\xi^{-1}(z))\big) \\
\nn &=& \omega(\omega(x,y),z).
\end{eqnarray}
By hypothesis $\xi(0)=0$, so $\xi^{-1}(0)=0$ and $$\omega(x,0)=\xi\circ\Phi(\xi^{-1}(x),0)=\xi\circ\xi^{-1}(x)=x,$$ which proves that $\omega(x,y)$ is indeed a group law. 
We conclude that $Z(p)$ is a group entropy.
\end{proof}

\section{Examples: New multiparametric group entropies}

Using the theoretical framework developed above, we shall construct new examples of group entropies associated with the multiplicative formal group law, which is the simplest nonadditive group law.  
To this aim, we shall consider
Tsallis entropy and Sharma-Mittal entropy, and we will combine them.
\subsection{Tsallis and Sharma-Mittal entropies}
Tsallis entropy was introduced in \cite{Tsallis1} as a generalization of Boltzmann-Gibbs entropy for long-range weakly chaotic systems (see \cite{Tbook} for a general discussion). It reads
\begin{equation}
S_{q}:= \frac{1-\sum_{i} p_{i}^{q}}{q-1}.
\end{equation}

An interesting generalization of Tsallis entropy is the bi-parametric entropy defined in 1975 by Sharma and Mittal (SM) in information theory. It has the form
\begin{equation}
S_{\alpha\beta}(p):=\frac{1}{1-\beta }\bigg[\bigg(\sum _{i=0}^W p_i^{\alpha }\bigg)^{\frac{1-\beta }{1-\alpha }}-1\bigg],
\end{equation}
where $\alpha>0$, and $\alpha,\beta\neq 1$. It is straightforward to show that for $\beta=\alpha $ we recover  Tsallis entropy and that when $\beta\to 1$ it reduces to 
R\'enyi's entropy \cite{Ren,JizAri}.
As is well known, both Tsallis and R\'enyi's entropies reduce to   Shannon's entropy in the limit when their parameter tends to $1$.

Assuming that the probability distribution satisfies:
$$p=(p_{ij}), \quad p_{ij}=p_i^Ap_j^B, \hspace{1mm} \foralll i,j $$
where $p^A\in \mathcal{P}_W,p^B\in P_{W'}$, then:
\begin{equation}\label{SM}
S_{\alpha\beta}(p)=S_{\alpha\beta}(p^A)+S_{\alpha\beta}(p^B)+(1-\beta)S_{\alpha\beta}(p^A)S_{\alpha\beta}(p^B)
\end{equation}
The Sharma-Mittal entropy is therefore a group entropy, whose composability law is given by the multiplicative formal group, as in the case of Tsallis entropy:
\begin{equation} \label{multip}
\Phi(x,y)=x+y+(1-q)xy=:x\oplus_q y .
\end{equation}
The above function is also known in the literature as the \emph{q-sum} \cite{Tbook}.

SM entropy has also been studied in the context of statistical mechanics \cite{FraPlas, Masi},  in particular for the description of anomalous diffusion phenomena \cite{FraDaff}.

We will consider a very simple case of composition of entropies, when $\xi= id$ and $m=1$. We propose two examples.

\subsection{Combining two copies of Sharma-Mittal entropy} We define the new tri-parametric group entropy 
$$Z_{\alpha_1,\alpha_2,\beta}(p)=S_{\alpha_1,\beta}\oplus_\beta S_{\alpha_2,\beta}.$$ 
Explicitly, it reads
\begin{equation} \label{NE1}
Z_{\alpha_1,\alpha_2,\beta}(p):= \frac{1}{\beta-1}\bigg[1-\bigg(\sum_{i=1}^{W} p_i^{\alpha_1}\bigg)^{\frac{\beta-1}{\alpha_1 -1}} \bigg(\sum_{i=1}^{W} 
p_i^{\alpha_2}\bigg)^{\frac{\beta-1}{\alpha_2 -1}}\bigg].
\end{equation}
If $\beta\in(0,1)$, then the function $\zeta$ satisfies the hypotheses of Proposition \ref{prop:zetaS}. Since each copy $\{S_{\alpha_i,\beta}\}$ 
satisfies the same composition law \eqref{SM}, by the previous theory $S_{\alpha_1,\alpha_2,\beta}$ is again a group entropy, with group law given by 
\begin{equation}\label{SM}
Z_{\alpha_1,\alpha_2,\beta}(p)=Z_{\alpha_1,\alpha_2,\beta}(p^A)+Z_{\alpha_1,\alpha_2,\beta}(p^B)+(1-\beta)Z_{\alpha_1,\alpha_2,\beta}(p^A)Z_{\alpha_1,\alpha_2,\beta}(p^B) \ .
\end{equation}
In other words, the group theoretical structure is still given by the formal group law \eqref{multip}. A numerical analysis shows that the entropy \eqref{NE1} is concave,
for instance, for $\alpha_1,\alpha_2,\beta\in(0,1)$.
\subsection{Combining Tsallis and Sharma-Mittal entropies}
Another interesting possibility arises when we combine Tsallis and Sharma-Mittal entropies by means of Theorem \ref{prop:CompZeta}, choosing again for simplicity $\xi= id$ and $m=1$. 
We define,
$$Z_{\alpha,q}(p):=S_{\alpha,q}\oplus_q S_{q}$$ 
obtaining,
\begin{equation}\label{NE2}
Z_{\alpha,q}=\frac{1}{q-1}\left[1-\bigg(\sum_{i=1}^{W}p_{i}^{q}\bigg)\bigg(\sum_{i=1}^{W}p_{i}^{\alpha}\bigg)^{\frac{q-1}{\alpha-1}}\right].
\end{equation}
This entropy satisfies the composition law 
\begin{equation}\label{Zmult}
Z_{\alpha,q}(p):=Z_{\alpha,q}(p^A)+Z_{\alpha,q}(p^B)+(1-q)Z_{\alpha,q}(p^A)Z_{\alpha,q}(p^B).
\end{equation}
namely, the group theoretical structure is the one defined by the formal group law \eqref{multip}. As in the previous example, one can ascertain that the entropy \eqref{NE2} is concave when 
$\alpha, q \in (0,1)$. 
\subsection{Entropies with more general composition law}
We could also assume $\xi\neq id$ in both examples with $\xi(0)=0$. Then, Proposition \ref{prop:zetaS} is satisfied, so it is an entropic function. Furthermore, by Theorem \ref{Teo:CompShf} 
the group law will be: $$\omega(x,y)=\xi(\xi^{-1}(x)+\xi^{-1}(y)+(1-\beta)\xi^{-1}(x)\xi^{-1}(y)) \ . $$

\section{Classical Information Geometry: a brief review}
First, we shall review some basic concepts in information geometry to fix the notation and make precise the definitions adopted. For a thorough exposition concerning this rapidly evolving field, we suggest to consult the monographs \cite{AmaMe,AmaIG}.
\begin{defin}[Parametric model]
Consider a family $M$ of probability distributions on $\mathcal{A}\subseteq \mathbb{R}^{k}$. We shall suppose that each distribution of $M$ may be parametrized using $n$ real-valued
variables $[\xi_1,\cdots,\xi_n]$ belonging to an auxiliary Euclidean space $\Xi\subset \R^n$. Therefore we can identify
\[
M=\{ p_\xi=p(x;\xi) \ | \ \xi=[\xi_1,\cdots,\xi_n]\in \Xi \},
\]
where $p(x;\xi)$ is the distribution function on $\mathcal{A}$ and $\xi\rightarrow p_\xi$ is injective. We say that $\Xi$ is a \emph{parametric space}, and $M$ a \emph{parametric model}.
\end{defin}

\begin{defin}[Statistical manifold]
A chart on $M$ is an application $\varphi:M\rightarrow \R^n$, $\varphi(p_\xi):=\xi$. If $\psi:\Xi\to\psi(\Xi)$ is a $C^\infty$ diffeomorphism, we introduce the new parameters $\eta:=\psi(\xi)$, so that
$M=\{p_{\psi^{-1}(\eta)} \ | \eta\in \psi(\Xi)\}$. This defines an atlas and the structure of a differentiable manifold, called a \emph{statistical manifold.}
\end{defin}
It is natural to consider geometric structures on a statistical manifold. Let us define some of them.
\begin{defin}[Fisher metric]\label{def:FishMet} Let $S$ be a statistical model. We introduce the metric
$$
g_{ij}(\xi) :=\int_{\mathcal{A}} \p_i l_\xi(x) \p_j l_\xi(x) p(x;\xi) dx,
$$
where $l_\xi(x):=\log p(x;\xi)$ and $\p_i=\frac{\p}{\p \xi^i }$, called the Fisher metric.
\end{defin}

It can be shown that this bilinear form is symmetric and positive semi-definite, and positive definite whenever $\{ \p_1l_\xi,\cdots,\p_n l_\xi\}$ are linearly independent.
We shall assume positive definiteness.

The inner product is given   on each fiber $T_\xi M$ by $g_\xi=\braket{\cdot,\cdot}_\xi$, with $g_\xi=g_{ij}(\xi)d\xi^i\otimes d\xi^j$. Precisely, if $X,Y\in T_\xi M$ we define
\[
\langle X, Y\rangle_\xi = g_{ij}d\xi^i\otimes d\xi^j(X^k \p_k, Y^l \p_l) =g_{ij}X^k \delta^i_k Y^l \delta^j_l=\E_\xi[(Xl_\xi)(Y l_\xi)]\ ,
\]
where $Xf:=\mathfrak{L}_X f$ and $\E_\xi$ denotes the mean value with $p_\xi$ as the distribution function.

Clearly, Fisher's metric endows a statistical manifold with the structure of a Riemannian manifold. We can also introduce a notion of duality in the class of metric connections.
\begin{defin}[Dual structure] Let $(M,g)$ be a Riemannian statistical manifold and let $\nabla,\nabla^*$ be affine connections on $M$.
We  say that the connections are dual if
$$\mathfrak{L}_Z\braket{X,Y}=\braket{\nabla_Z X,Y}+\braket{ X,\nabla^*_Z Y}, \qquad \forall X,Y,Z\in \T^1_0(M) \ . $$
In this case we shall say that ($g,\nabla,\nabla^*$) is a \emph{dual structure} on $M$.
\end{defin}
\begin{defin}[Divergence] Let $M$ be a statistical manifold and suppose that we are given a smooth function $D( \cdot|| \cdot) : M \times M \to \mathbb{R}$ satisfying the following properties
$\forall$ $p,q \in M$:
\begin{itemize}
\item[i)] $D(p||q)\ge 0$, and $D(p||q)=0$ iff $p =q$,
\item[ii)] $D[\partial_i\partial_j||\cdot]$ is positive definite, where $D[\partial_i\partial_j||\cdot](p):=\partial_i\partial_j D(p(\xi)||p(\xi'))\vert_{\xi'=\xi}$,
$\partial_i:=\partial/\partial \xi_i$ and $p=p(\xi)$.
\end{itemize}
Then, $D$ is said to be a \emph{divergence.}
 \label{def:Div}
\end{defin}
\begin{nota}
Condition i) guarantees positive semi-definiteness of $D$. If only this condition holds, we will talk about a \emph{semi-divergence}.
In general, a divergence is a pseudo-distance, i.e. it does not in general satisfy the other axioms of distance (although it could be symmetrized).
\end{nota}
Given a divergence, the following geometric structures are defined:
\begin{defin}\label{def:EstDiv} Let $D$ be a divergence on $M$. We define the following functions:
  \begin{equation}  \label{met} \braket{X,Y}^{(D)}:=-D[X||Y]=D[XY||\cdot], \quad  X,Y \in TM \end{equation}
 \begin{equation}  \label{af1} \braket{\nabla^{(1,D)}_X Y, Z}^{(D)}:=-D[XY||Z], \quad  X,Y,Z \in TM \end{equation}
\begin{equation}  \label{af2} \braket{\nabla^{(2,D)}_X Y, Z}^{(D)}:=-D[Z||XY], \quad  X,Y,Z \in TM  \end{equation}
where by $D[X||\cdot](p)$ we mean $X^iD[\p_i||\cdot](p)$
\end{defin}
They have interesting properties:
\begin{teo} Let $D$ be a divergence on $M$. Then, the previous structures represent a metric (eq. \eqref{met}) and  two affine connections (eqs. \eqref{af1}--\eqref{af2}), respectively.
In addition, ($M, g^{D}, \nabla^{(1,D)}, \nabla^{(2,D)}$) is a dual structure on $M$.
\end{teo}

\section{Divergence Functions from group entropies}\label{sec:DivFun}
The aim of this section is to give an information-theoretical interpretation of the previous results concerning entropic functions and  group entropies. More precisely, we shall prove that the composition procedure
allowing us to construct new group entropies from old ones, also permits to generate in a similar way new families of divergence functions from old ones. Their geometric properties will be studied in detail below.

\vspace{1mm}

First, it is straightforward to show the following preliminary result, which allows us to construct more general
semi-divergences using previously defined ones.
\begin{prop}\label{prop:zetaD} Let $\{D_i\}_{i=1}^m$ be a family of semi-divergences. Let $\zeta:\Omega\subset \R^m\to\mathbb{R}$, with $\R_{>0}^n\cup\{\boldsymbol{0}\}\subset \Omega$ be a function 
such that
$$\zeta\vert_{\R_{>0}^n\cup\{\boldsymbol{0}\}}:\R_{>0}^n\cup\{\boldsymbol{0}\}\to\R_{\ge 0} \hspace{2mm}\text{and} \hspace{2mm}\zeta\vert_{\R_{>0}^n\cup\{\boldsymbol{0}\}}(\textbf{x})=0
\Longleftrightarrow \textbf{x}=0 \ . $$
\noindent Then the function
\begin{equation} \label{eq:Dsemid}
D(p||q):=\zeta(D_i(p||q))
\end{equation}
is a semi-divergence.
\end{prop}
\begin{ex}[Taylor polynomial for divergences]\label{ex:taylorD} If $D_j$ is a semi-divergence, we can define the following multi-parametric semi-divergence:
$$D(p||q)=\sum_{i=1}^n\sum_{\sum i_j=i}\alpha^i_{i_1...i_n}\prod_{j=1}^m(D_j(p||q))^{i_j},$$
where $\alpha^i_{i_1...i_n}\ge 0$ (and at least one different from 0) and we have used Proposition \ref{prop:zetaD}.
\end{ex}
We study now the  geometry associated with the semi-divergence \eqref{eq:Dsemid} constructed above.

\begin{teo}\label{Teo:GeoZeta} Let $D(p||q)=\zeta(D_i(p||q))$ be a semi-divergence, where  $\{D_k\}_{k=1}^m$ is a family of divergences, and the function  $\zeta$, defined as in Proposition 
\ref{prop:zetaD}, is of class $C^{3}(\Omega)$. Then, $D$ is a divergence if and only if  $\boldsymbol{0}$ is not a critical point for $\zeta$. Furthermore,
\begin{itemize}
\item[i)] $g_{ij}=\sum_{k=1}^m\tilde{\partial_k} \zeta(\boldsymbol{0})g_{ij}^{(k)}$,
\item[ii)] $\Gamma_{ij,k}=\sum_{k=1}^m\tilde{\partial_l} \zeta(\boldsymbol{0})\Gamma_{ij,k}^{(l)}$,
\item[iii)]$\Gamma_{ij,k}^*=\sum_{k=1}^m\tilde{\partial_l} \zeta(\boldsymbol{0})\Gamma_{ij,k}^{(l)*}$ \ .
\end{itemize}
Here $\tilde{\partial_k}f(x):=(\partial f)/(\partial x^k)$, whereas $\partial_i$ is the partial derivative with respect to the parameters representing the coordinates of the statistical manifold.
\end{teo}
\begin{proof}
Let us define $\boldsymbol{D}(\xi||\xi'):=(D_1(p_\xi||p_{\xi'}),...,D_m(p_\xi||p_{\xi'}))$. Using the chain rule, we have
$$\partial_i \zeta(\boldsymbol{D}(\xi||\xi'))=\sum_k\tilde{\partial_k}\zeta(\boldsymbol{D}(\xi||\xi')) \partial_i D_k(\xi||\xi'). $$
Then
\begin{equation}\label{eq:FunDivMet}
\p_j\p_i( \zeta(\boldsymbol{D}(\xi||\xi'))) = \sum_{k,l}\tilde{\partial_{lk}^2}\zeta(\boldsymbol{D}(\xi||\xi'))\partial_j D_k \partial_i D_k+\sum_k \tilde{\partial_{k}}\zeta(\boldsymbol{D}(\xi||\xi'))\partial_{ji}^2 D_k.
\end{equation}
Now, let us evaluate the previous expression at $\xi=\xi'$. By Definition \ref{def:Div}, it is clear that $(\partial_i D_k)\vert_{\xi'=\xi}=\boldsymbol{0}$ $\foralll k\in\{1,...,m\}$. 
Consequently, the first term vanishes. For the second term, take into account that $D_k(\xi||\xi)=0$ by definition, so it converts into
$\sum_k \tilde{\partial_{k}}\zeta(\boldsymbol{0})\partial_{ji}^2 D_k\vert_{\xi'=\xi} $.
Then, using the definition of $g_{ij}$ for each divergence (see Definition \ref{def:EstDiv}), the result i) follows. \\
Let us prove that the metric $g$ is positive definite. Observe that by Proposition \ref{prop:zetaD}, the restricted function $\zeta^R:=\zeta\vert_{\R_{>0}^n\cup\{\boldsymbol{0}\}}$ is such that 
$\zeta^R(\boldsymbol{0})=0$ and $\zeta^R(\boldsymbol{x\neq 0})>0$. So, consider the directional derivatives along $e\in\R^n_{>0}$. As $\zeta$ is $C^1$, 
the right derivative ($t>0$, $\zeta^R(\boldsymbol{x\neq 0})-\zeta^R(\boldsymbol{0})>0$), non-negative)   equals the derivative, i.e.,
$$\frac{d^+}{dt}\bigg\vert_{t=0}\zeta(e t)=\frac{d}{dt}\bigg\vert_{t=0}\zeta(e t)=\sum_i\pt_i\zeta(\boldsymbol{0}) e^i\geq 0. $$
Since by hypothesis $\boldsymbol{0}$ is not a critical point for $\zeta$, then $\existss \pt_i\zeta(\boldsymbol{0})$ not vanishing. 
If $\existss \mathcal{I}$ such that $ \pt_i\zeta(\boldsymbol{0})<0$ when $i\in\mathcal{I}$, 
define $e_l\in\R^+_{>0}$ such that $e^i_l:=l$ if $i\in\mathcal{I}$ and $e^i_l:=l^{-1}$ if $i\notin\mathcal{I}$. Then
$$\lim_{l\to\infty}\sum_i\pt_i\zeta(\boldsymbol{0}) e^i=-\infty, $$
thus $\existss l\in\mathbb{N}$ such that $\sum_i\pt_i\zeta(\boldsymbol{0})e^i<0$, but it must be non-negative. 
Therefore, $\mathcal{I}=\emptyset$, i.e., $\pt_i\zeta(\boldsymbol{0})\geq 0$ and at least one of these derivatives is positive. Then
$$\braket{v,v}_{g}:=\sum_{ij}g_{ij}v^iv^j=\sum_{k}\pt_k\zeta(\boldsymbol{0})\braket{v,v}_{g_k}\geq 0, $$
because by assumptions $g_k$, the metrics associated with the divergences $D_{k}$ are positive definite  and the coefficients are non-negative. Equality holds if and only if
 $\pt_k\zeta(\boldsymbol{0})\braket{v,v}_{g_k}=0$ $\foralll k$, but let $k'$ such that $\pt_{k'}\zeta(\boldsymbol{0})>0$ (we have shown above that such $k'$ exists), so $\braket{v,v}_{g_{k'}}=0$. 
By positive definiteness of $g_{k'}$, $v=\textbf{0}$; thus we conclude that $g$ is positive definite.\\
Regarding the connection symbols, just note that if we differentiate eq. \eqref{eq:FunDivMet} with respect to $\xi'_k$, any summand which includes one partial derivative of $D$ will vanish, so the only term will be (after evaluating at $\xi=\xi'$):
$$\sum_l \tilde{\partial_{l}}\zeta(\boldsymbol{0})\partial_{ji}^2\p'_k D_l\vert_{\xi'=\xi}, $$
and then, the result for $\Gamma_{ij,k}$ follows. A similar derivation can be done for $\Gamma_{ij,k}^*$.
\end{proof} 
\begin{nota} Notice that if we differentiated again, that derivative would \emph{not} be a linear combination of the derivatives of $D_k$ (as it is for the second and third derivative, i.e., 
for $g$ and $\Gamma$), new non-linear terms would arise.
\end{nota}

\begin{cor}\label{cor:GeoZeta} 
The new connection is given by:
\begin{equation}
\Gamma^m_{ij}(\xi)=\sum_{l,n}A^m_{l,n}(\xi)\Gamma^{n(l)}_{ij}(\xi)\ , 
\end{equation}
where $A^m_{l,n}(\xi):=\sum_{k,n}\p_l\zeta(\textbf{0})g^{mk}(\xi)g_{kn}^{(l)}(\xi)$.
\end{cor}
\begin{proof}
Indeed, by Definition \ref{def:EstDiv} 
\[\braket{\nabla_{\p_i}\p_j,\p_k}^D=\Gamma_{ij,k} \quad\text{and}\quad \braket{\nabla_{\p_i}\p_j,\p_k}^D=\braket{\Gamma_{ij}^l\p_l,\p_k}^D \ ,
\] 
so we have
\begin{equation} \label{Gammalij}
\Gamma^l_{ij}g_{lk}=\sum_l\partial_l\zeta(\textbf{0})g^{(l)}_{kn}\Gamma^{n(l)}_{ij}, 
\end{equation}
and then it easily follows if we multiply the previous expression \eqref{Gammalij} by $g^{km}$, the inverse of $g$, and sum it over $k$. Note that $g=\sum_l\pt_l\zeta(\boldsymbol{0})g^{(l)}$ is invertible because it is positive definite.

\end{proof}
\begin{ex} Consider the semi-divergence of Example \ref{ex:taylorD}. If the hypotheses of Theorem \ref{Teo:GeoZeta} are satisfied, the associated new metric will depend linearly on the previous metrics:
$g_{ij}=\sum_{k=1}^m\alpha^1_{k}g_{ij}^k $.
\end{ex}

\section{The $D_{{h,f}}$ divergence and its Geometry}\label{sec:Dhf}
In this Section, we will discuss the information geometry related to the class of $(h,f)$-entropies. The results obtained here can be seen as a particular case of the previous general approach; 
however, we will discuss them explicitly due to the considerable relevance of this class in many applicative contexts.
   
We will show that the SM divergence belongs to a more general family of divergences that will be defined below.
\begin{defin} Let $f, h$ be functions such that $f:\R_{\geq 0}\to \mathbb{R}$ is continuous and strictly concave (resp. convex), $h$ is continuous and strictly decreasing 
(resp. increasing) with $h(f(1))=0$. Let $\mathcal{P}^{+}_W$ denote the space of probability distributions $(p_0,p_1,\ldots,p_n)$, with $p_i>0$ $\forall i$.
For $p,q\in \mathcal{P}^{+}_W$ we introduce the $(h,f)$-divergence
\begin{equation} \label{eq:hfdiv}
D_{h,f}(p||q):=h\left(\sum_{i=0}^W q_if\left(\frac{p_i}{q_i}\right)\right) \ .
\end{equation}
\label{def:DgfG}
\end{defin}
An analogous expression holds for the continuous case.
The latter divergence can be considered as a generalization of the standard $f$-divergence \cite{AmaIG}. Apart the Kullback-Leibler divergence, the R\'enyi's and $\alpha$-divergences, another
interesting particular case in the class  \eqref{eq:hfdiv} is given, for instance, by the Sharma-Mittal divergence, whose geometry has been recently investigated by Nielsen et al. in \cite{Nielsen}. The SM-divergence is defined for $p,q\in\mathcal{P}_W$ by
\begin{equation}\label{form:Dab}
D_{\alpha\beta}(p||q):=\frac{\left(\sum _{i=0}^W p_i^{\alpha } q_i^{1-\alpha }\right){}^{\frac{1-\beta }{1-\alpha }}-1}{\beta -1} \ .
\end{equation}
\begin{teo} \label{Teo:Dsigma}Let $D_{h,f}(p||q)$ be the generalized divergence \eqref{eq:hfdiv}, where $f,h$ are twice differentiable functions, with $f$ strictly concave (resp. convex) 
and $h$ strictly increasing 
(resp. decreasing).
Then, the associated metric tensor  is given by\footnote{In accordance with Chentsov's theorem, \cite{AmaMe}.}
\begin{itemize}
\item[i)] Discrete case, for $p\in\emph{int}(\mathcal{P}_W)$

$g_{ij}=\left(\dfrac{\delta_{ij}}{p_j}+\dfrac{1}{p_0}\right)h '(f(1)) f''(1).$
\item [ii)] Continuous case,

$g_{ij}=h '(f(1)) f''(1)g^F_{ij},$

where $g_{ij}^F$ is the Fisher metric.
\end{itemize}
\end{teo}
\begin{proof} We have
$$\sum_i f(p_i/q_i)q_i \leq f\left(\sum_i p_i \frac{q_i}{q_i}\right)=f(1) $$
by Jensen's inequality. Then,
$$h\left(\sum_i f(p_i/q_i)q_i\right)\geq h(f(1))=0. $$
As $f$ is strictly concave, equality holds iff $p_i/q_i=p_j/q_j\equiv c$ $\foralll i,j$, but as $\sum_ip_i=\sum_iq_i=1$, then, $c=1$, i.e., $p_i=q_i$ $\foralll i$. 
If $f$ is strictly convex and $h$ is decreasing, the last inequality remains unchanged.
A direct calculation shows
$$D_{h,f}[\p^2_{ij}||\cdot]=f''(1) h '(f(1)) \int_\Omega \frac{\p_ip(x,\xi)\p_j p(x,\xi)}{p(x)} \, dx. $$
Now, using $\p_ip(\xi,x)=\p_i\log p(\xi,x)p(\xi,x)$ the result follows, see Definition \ref{def:FishMet}. In the discrete case, the calculation is straightforward, 
now take into account that $p_0=1-\sum _{i=1}^W p(i)$ and general conditions for probability distributions. The result is the one stated above, discrete Fisher metric. 
To prove it is  positive definite, note that
$$g_{ij}v^iv^j=\sum_{i=1}^W \dfrac{(v^i)^2}{p_i}+\dfrac{(\sum_{i=1}^W v_i)^2}{p_0}\geq 0 .$$
And if equality holds, $(v^i)^2=0$ $\forall i$, so it is positive definite.
\end{proof}

The dual connections associated with the $(h,f)$-divergences are determined in the following
\begin{teo} \label{Teo:DsigmaCris}Let $D_{h,f}(p||q)$ be the generalized divergence \eqref{eq:hfdiv}, where $f,h$ are three times differentiable functions. 
Then, its associated dual connections (Definition \ref{def:EstDiv}) are given by:
$$\Gamma_{ij,k}^{(h,f)}=c\Gamma_{ij,k}^{(-\alpha)}, \quad \Gamma_{ij,k}^{(h,f) *}=c\Gamma_{ij,k}^{(\alpha)}, $$
where $c:=h '(f(1)) f''(1)$, $\alpha=(2 f'''(1)+3 f''(1))/f''(1)$ and:
{\small $$\Gamma^{(\alpha)}_{ij,k}(\xi):= \E_\xi\left[ \left(\p_i\p_j l_\xi+\frac{1-\alpha}{2}\p_i l_\xi \p_j l_\xi\right)(\p_k l_\xi)\right].$$}
This implies: $\nabla^{(h,f)}=\nabla^{(-\alpha)}$, $\nabla^{(h,f)*}=\nabla^{(\alpha)}$ (the $\alpha$-connection, \cite{AmaMe})
\end{teo}
\begin{proof}
A direct, but lengthy calculation shows 
\begin{align*}
&\Gamma_{ij,k}^{(h,f)}=h'(f(1))f''(1) \left(\int_\Omega \frac{\p_{ij}^2 p_\xi(x) \p_kp_\xi(x)}{p(x)} \, dx+\right.\\
&\left.\dfrac{f'''(1)+f''(1)}{f''(1)} \int_\Omega \frac{\p_i p_\xi (x)\p_j p_\xi (x)\p_k p_\xi (x)}{p(x)^2} \, dx\right).
\end{align*}
Now, we use the fact that $\p_ip_\xi=\p_i\l_\xi(x)p_\xi$ and $\p_{ij}^2 l_\xi=-p^{-2}_\xi\p_i p_\xi\p_j p_\xi+p^{-1}_\xi\p^2_{ij}p$ to finally get:
$$c\cdot\E_\xi\left[ \left(\p_i\p_j l_\xi+\left(\dfrac{f'''(1)}{f''(1)}+2\right)\p_i l_\xi \p_j l_\xi\right)(\p_k l_\xi)\right]$$
and solving $\dfrac{f'''(1)}{f''(1)}+2=\dfrac{1-\alpha}{2}$ the result follows. A similar calculation can be repeated for the other set of symbols. The last statement follows from the definitions, i.e.,
$$\Gamma^k_{ij}=g^{kl}c\Gamma_{ij,l}^{-\alpha}=c^{-1}g^{kl}_Fc\Gamma_{ij,l}^{-\alpha}=\Gamma^{k(-\alpha)}_{ij} \ . $$ 
\end{proof}

\section*{Acknowledgements}
M.A.R.  has been partly supported by the research project FIS2015-63966, MINECO, Spain. P.T. has been partly supported by the research project FIS2015-63966, MINECO, Spain,
by the ICMAT Severo Ochoa grant SEV-2015-0554 and by the GNFM, Italy.

\end{document}